\documentclass[conference]{IEEEtran}

\IEEEoverridecommandlockouts

\ifCLASSINFOpdf
  \usepackage[pdftex]{graphicx}
\else
  \usepackage[dvips]{graphicx}
\fi
\usepackage{graphics} 
\usepackage{epsfig} 
\usepackage{times} 
\usepackage{amsmath} 
\usepackage{amssymb}  
\usepackage{cite}
\usepackage{algorithm}
\usepackage{algorithmic}
\usepackage[table]{xcolor}
\usepackage{bbm}
\usepackage{caption}
\usepackage{subcaption}
\usepackage{tikz}

\definecolor{myemphcolor}{cmyk}{0.57,0.55,0,0}

\newcommand\subscr[2]{#1_{\textup{#2}}}

\definecolor{black}{rgb}{0,0,0}
\definecolor{Red}{rgb}{1,0,0}
\definecolor{Blue}{rgb}{0,0,1}
\definecolor{Green}{rgb}{0,1,0}
\definecolor{magenta}{rgb}{1,0,.6}
\definecolor{lightblue}{rgb}{0,.5,1}
\definecolor{lightpurple}{rgb}{.6,.4,1}
\definecolor{gold}{rgb}{.6,.5,0}
\definecolor{orange}{rgb}{1,0.4,0}
\definecolor{hotpink}{rgb}{1,0,0.5}
\definecolor{newcolor2}{rgb}{.5,.3,.5}
\definecolor{newcolor}{rgb}{0,.3,1}
\definecolor{newcolor3}{rgb}{1,0,.35}
\definecolor{darkgreen1}{rgb}{0, .35, 0}
\definecolor{darkgreen}{rgb}{0, .6, 0}
\definecolor{darkred}{rgb}{.75,0,0}

\xdefinecolor{olive}{cmyk}{0.64,0,0.95,0.4}
\xdefinecolor{purpleish}{cmyk}{0.75,0.75,0,0}

\newcommand{\mA}{\mathcal{A}}

\newcommand{\mN}{\mathcal{N}}

\newcommand{\mM}{\mathcal{M}}

\newcommand{\mK}{\mathcal{K}}

\newtheorem{theorem}{\bf Theorem}

\newtheorem{corollary}{\bf Corollary}

\newtheorem{definition}{\bf Definition}
\newtheorem{remark}{\bf Remark}

\newcommand{\longthmtitle}[1]{\mbox{}\textup{\textbf{(#1):}}}
\newcommand{\real}{{\mathbb{R}}}
\newcommand{\realpositive}{{\mathbb{R}}_{>0}}

\newcommand{\realnonnegative}{{\mathbb{R}}_{\ge 0}}

\newcommand{\oprocendsymbol}{\hbox{$\bullet$}}
\newcommand{\oprocend}{\relax\ifmmode\else\unskip\hfill\fi\oprocendsymbol}

\newcommand{\proofendsymbol}{\hbox{$\Diamond$}}
\newcommand{\proofend}{\relax\ifmmode\else\unskip\hfill\fi\proofendsymbol}

\newcommand{\Rmnum}[1]{\expandafter\@slowromancap\romannumeral #1@}

\begin{document}

\title{Context-Aware Wireless Small Cell Networks: How to Exploit User Information for Resource Allocation}
\author{
\IEEEauthorblockN{Ali Khanafer\IEEEauthorrefmark{1}, Walid Saad\IEEEauthorrefmark{2}, and Tamer Ba\c{s}ar\IEEEauthorrefmark{1}}
\IEEEauthorblockA{\IEEEauthorrefmark{1}Coordinated Science Laboratory, University of Illinois at Urbana-Champaign, USA, Email: \{khanafe2,basar1\}@illinois.edu}
\IEEEauthorblockA{\IEEEauthorrefmark{2}Wireless@VT, Bradley Department of Electrical and Computer Engineering, Blacksburg, VA, Email: walids@vt.edu}
\thanks{This research supported in part by an AFOSR MURI Grant FA9550-10-1-0573 and by the U.S. National Science Foundation under Grants CNS-1443159 and CNS-1513697.}
}

\maketitle

\begin{abstract}
In this paper, a novel context-aware approach for resource allocation in two-tier wireless small cell networks~(SCNs) is proposed. In particular, the SCN's users are divided into two types: frequent users, who are regular users of certain small cells, and occasional users, who are one-time or infrequent users of a particular small cell. Given such \emph{context} information, each small cell base station (SCBS) aims to maximize the overall performance provided to its frequent users, while ensuring that occasional users are also well serviced. We formulate the problem as a noncooperative game in which the SCBSs are the players. The strategy of each SCBS is to choose a proper power allocation so as to optimize a utility function that captures the tradeoff between the users' quality-of-service gains  and the costs in terms of resource expenditures. We provide a sufficient condition for the existence and uniqueness of a pure strategy Nash equilibrium for the game, and we show that this condition is independent of the number of users in the network. Simulation results show that the proposed context-aware resource allocation game yields significant performance gains, in terms of the average utility per SCBS, compared to conventional techniques such as proportional fair allocation and sum-rate maximization.
\end{abstract}
\section{Introduction}
The continuously increasing demand for bandwidth-intensive wireless services mandates major structural changes to today's wireless cellular networks~\cite{CISC}. One key change is in the deployment of low-power, low-cost, small cell base stations~(SCBSs), overlaid on existing cellular infrastructure~\cite{IG00,JA03}. Such small cells, operating in the licensed spectrum and deployed both indoors (femtocell at home) and outdoors (operator-deployed picocells and microcells), allow for reducing the distance between users and their serving stations, thus resulting in a dramatic increase in the wireless capacity. It is expected that the next generation of wireless cellular networks will consist of a dense and widespread deployment of small cells, varying in range and capabilities.

Despite its immense technological benefits, the deployment of small cell wireless networks introduces numerous technical challenges at different levels such as resource management~\cite{DL02,RH02,EE05,MOB00}, self-organization~\cite{GTF00,GTF02,GTF04}, and interference management~\cite{SK01,JA05,EE00}. In \cite{DL02}, the authors discuss the use of various coordination techniques for mitigating interference while focusing on 3GPP standard-oriented approaches. The authors in \cite{RH02} propose a joint power and subcarrier resource allocation scheme for small cell networks in which users are prioritized based on their quality-of-service (QoS) needs. The optimal allocation of traffic between different, heterogeneous spectrum bands using a multi-mode femtocell is studied in \cite{EE05}. The work in \cite{MOB00} studies the mobility enhancements needed to maintain desirable performance in small cell networks. A distributed game-theoretic approach for interference management is presented in \cite{GTF00}, while \cite{GTF02} develops a hierarchical self-organizing approach for power allocation in two-tier small cell networks. Other interesting aspects of resource allocation in small cell networks are studied in \cite{GTF04,SK01,JA05,EE00,HC05}.

Although they present interesting approaches for handling resource allocation in small cell networks, most of these existing works optimize the system performance under the assumption that only physical layer metrics, such as channel gain and power capabilities, are necessary for such optimization. In practice, given the proximity of the small cells to their users and the density of small cell deployment, the small cell network can harness a large amount of information on each user; thus, they can provide a better approach to resource allocation. In particular, a \emph{context-aware} resource management mechanism is seen as a suitable and enabling approach for exploiting new information on the users -- such as smartphone type or even social metrics -- to better allocate resources in the small cell network. While the notion of context-aware resource allocation has been studied in conventional cellular systems such as in \cite{CA00,SV00,SV01}, these approaches are not tailored to the specifics of small cell networks. Moreover, unlike such existing works~\cite{CA00,SV00,SV01}, this paper will exploit previously unexplored context information to optimize small cell networks.

The main contribution of this paper is its development of a novel, context-aware resource allocation scheme for small cell networks that exploits its historical knowledge on the users to predict their resource requirements and optimize the overall resource allocation performance. In particular, in dense small cell networks, where each small cell base station (SCBS) covers a relatively small geographical area, the SCBSs can be very effective in predicting the resource usage patterns of their \emph{frequent} users. A frequent user is a user who has been regularly visiting a given small cell over time. An example of such a user is a student who uses the same campus SCBSs every day. Here, the SCBS can monitor the data requests of such users and, then, exploit this information to proactively allocate resources to them. In essence, deploying SCBSs with limited coverage allows for a more focused view of the network, as each SCBS needs only to interact with a limited number of users as opposed to the classical cellular networks where the macro-cell base station handles a massive number of users. 

We formulate this context-aware resource allocation problem as a noncooperative game, and we provide a sufficient condition for the existence and uniqueness of the equilibrium. Further, we provide a method for computing this equilibrium. In this game, the SCBCs optimally allocate their powers so as  to optimize the tradeoff between maximizing the data rate to all users (frequent and non-frequent) and minimizing the misspending of resources  based on the knowledge of the demands of frequent users. Subsequently, our results show that, by endowing the SCBSs with such an ability  to predict and learn the usage patterns, we can significantly boost the performance of the network and lead to improved rates, compared to classical techniques such as sum-rate maximization or proportionally fair allocation.

The rest of this paper is organized as follows: Section~\ref{SARA} presents the proposed context-aware system model. In Section~\ref{sec:game}, we formulate the problem as a noncooperative game and study the game's properties. Simulation results are presented and analyzed in Section~\ref{sec:sim}. Finally, conclusions are drawn in Section~\ref{sec:conc}. An Appendix contains the proof of the main result of the paper.

\section{Context-Aware System Model} \label{SARA}
From the perspective of each SCBS, we can categorize user equipments (UEs) into two main groups: \emph{frequent} and \emph{occasional} users. Here, frequent users are those users who regularly access a certain SCBS. For these frequent UEs, the SCBSs are able to predict the resources required to serve them. As a result, we consider that each SCBS will strive to meet the demand requirements of its frequent users (which it learned over time). In contrast, occasional users are those who make sporadic uses of the SCBSs. For example, they can be seen as mobile users who are passing by the coverage area of an SCBS at one time, but do not often return there. Traffic patterns of mobile UEs can vary considerably based on the type of services requested. Indeed, UEs differ in the way they utilize the spectrum; while some use their smartphones for voice mainly, others use them for social media and video streaming. For frequent users, the SCBSs can generally observe the regular traffic and usage patterns. Consequently, we define the notion of \emph{context-aware} resource allocation as \emph{the ability of SCBSs to learn the traffic patterns of their frequent users and, hence, be able to service them better}.

In the network considered, each SCBS attempts to maximize the sum-rate it provides to all UEs served (frequent and occasional). Naturally, the provided rate depends on the amount of power allocated which, in turn, depends on each user's QoS requirements. For frequent users, as the SCBS knows their traffic patterns ahead of time, we assume that each SCBS knows the amount of resources required by its frequent UEs. For the occasional users, such an assumption does not hold since occasional UEs have limited interaction with the SCBSs and can be one-time visitors.

For each UE, an SCBS must provide the exact amount of required resources, to meet the QoS needs. On the one hand, if an SCBS provides a UE with an insufficient amount of resources, then the UE will not experience the required QoS. On the other hand, even though providing a UE with more resources than it requests can boost the overall sum-rate, the scarcity of the spectrum imposes that each SCBS utilizes its resources wisely and sparingly. For example, instead of spending extra power on an occasional UE that did not request it, the SCBS can exploit this additional resource to improve the performance of its regular, frequent users. To discourage such \emph{misspending} of resources, next, we will propose a resource allocation mechanism in which the SCBSs optimizes the tradeoff between maximizing the overall sum-rate and minimizing the cost incurred by misspending resources.

The broadcast nature of wireless networks leads to interference between the SCBSs in the downlink. Therefore, the choice of a resource allocation scheme at one SCBS will determine the interference level experienced on other, neighboring SCBSs. Moreover, in small cells, it is of interest to develop \emph{self-organizing} resource allocation mechanisms in which each SCBS optimizes its own, individual objective, without relying on a centralized control and with little coordination with its neighboring SCBSs. Due to this \emph{interdependence} of the resource allocation decisions due to factors such as \emph{mutual interference} and the need for self-organization, we will formulate the resource allocation problem as a strategic, noncooperative game~\cite{NSF1}.

\section{Context-Aware Resource Allocation Game: Formulation and Solution}\label{sec:game}
Consider the downlink of a two-tier small cell network composed of $M$ SCBSs and $N$ UEs.  We let $\mathcal{M}$ be the set of all SCBSs and $\mathcal{N}$ be the set of all UEs.  Let $\mathcal{N}_1^i$, $\mathcal{N}_2^i$ be the sets of frequent and occasional users served by SCBS $i$, where $\mathcal{N}_1^i \cap \mathcal{N}_2^i = \emptyset$, and $\mathcal{N}_1^i, \mathcal{N}_2^i \subseteq \mathcal{N}$, for all $i$. Also, define $\mN_i = \mN_1^i \cup \mN_2^i$, $N_i := |\mathcal{N}_i|$, $N_1^i := |\mathcal{N}_1^i|$, and $N_2^i := |\mathcal{N}_2^i|$. Note that a given UE could be categorized as frequent by one SCBS and occasional by another. Also, note that SCBS $i$ can serve a strict subset of $\mN$ and does not necessarily serve all users in the network. Formally, we allow for scenarios where $\mN_i \subset \mN$.

We assume that the SCBSs share the same spectrum and have access to the same $K$ subcarriers (or channels). We denote the set of the subcarriers by $\mK$. We assume that each SCBS performs subcarrier allocation using an efficient allocation technique from the literature \cite{IG00,RH02,lopez2009ofdma}. Two SCBSs will interfere if and only if they allocate power to a given user over the same subcarrier. Further, we assume that there is no inter-user interference on a given subcarrier. In other words, we assume that each subcarrier can be used by only one UE. Since the number of users can be very large compared to $K$, we assume that $\mN$ is the set of users selected (out a larger population) by the subcarrier allocation at the SCBSs. To this end, we assume that $K \geq N$. Our setting can be extended to allow for inter-user interference (hence, allowing for $K < N$), but we restrict our attention to the no inter-user interference case in this paper in order to demonstrate the idea of context-aware resource allocation.

Let $\mK_i \subseteq \mK$ denote the set of subcarriers available to SCBS $i$, and let $K_i := |\mK_i|$. Note that with these definitions SCBS $i$ is allowed to use a strict subset of the subcarriers only, i.e., we allow for scenarios where $\mK_i \subset \mK$. We introduce the mapping $\pi_i:\mK_i \to \mN_i$, where $\pi_i(k)\in \mN_i$ is the user to be served over subcarrier $k$, according to the subcarrier allocation performed at SCBS $i$. By our assumptions, the mapping $\pi_i$ is one-to-one and onto. 

Let $u_{ij}[k] \in \realnonnegative$ be the power allocated by SCBS $i \in \mathcal{M}$ to UE $j \in \mathcal{N}$ over subcarrier $k$. We denote the power allocation matrix of SCBS $i$ by $U_i \in \mathbb{R}^{N \times K}$, where $[U_i]_{jk} = u_{ij}[k] \in \realnonnegative$. Each SCBS has a finite power budget $\subscr{P}{max}$, and we assume that
\[
\sum_{k=1}^K u_{i\pi_i(k)}[k]  = \sum_{k=1}^K \sum_{j \in \mathcal{N}_i} u_{ij}[k] \leq \subscr{P}{max},
\]
for all $i\in \mM$. It is possible to extend this setting to the case where each SCBS has its own power constraint $\subscr{P}{max}^i$; here, we assume that $\subscr{P}{max}^i = \subscr{P}{max}$, for all $i$, for simplicity. Note that under the no inter-user interference assumption, each column of $U_i$ can contain only one nonzero entry, for all $i$. Also, SCBSs $i$ and $j$ will interfere if and only if $u_{il}[k]>0$ and $u_{jl}[k]>0$ for some $l \in \mN$ and some $k \in \mK$. We assume that $u_{ij}[k] > 0$ if and only if $k \in \mK_i$. We can now formally define the action set of SCBS $i$ as
\begin{eqnarray}\label{eqn::actionSet}
\mA_i \hspace{-2mm} & = & \hspace{-2mm}\left\{U_i \in \realnonnegative^{N_i \times K_i} : [U_i]_{jk} = u_{ij}[k]; \mathbf{1}^TU_i \mathbf{1} \leq \subscr{P}{max};  \right. \nonumber \\
&&\hspace{-2mm} \left. u_{ij}[k] > 0 \iff j= \pi_i(k) \text{ for } k\in \mK_i, j \in \mN_i    \right\}.
\end{eqnarray}
Note that the action sets depend on the subcarrier allocation implicitly.

Given the above, we can now write the downlink signal-to-interference-plus-noise ratio (SINR) achieved at user $j$ due to the power allocated to it by SCBS $i$ over subcarrier $k$ as follows:
\begin{equation*}
\text{SINR}_{ij}[k] = \frac{|h_{ij}[k]|^2d_{ij}^{-\alpha}u_{ij}[k]}{  \sum\limits_{ \substack{l \in \mM \\ l \neq i} } |h_{lj}[k]|^2d_{lj}^{-\alpha}u_{lj}[k] + \sigma_j^2[k]},
\end{equation*}
where $d_{ij} \in \realpositive$ is the distance from SCBS $i$ to user $j$, $h_{ij}[k] \in \real$ is the channel gain from SCBS $i$ to user $j$ over subcarrier $k$, and $\sigma_j^2[k] \in \realnonnegative$ is the variance of the additive zero-mean circular complex Gaussian noise at UE $j$ over subcarrier $k$. Consequently, we can write the rate allocated by SCBS $i$ to UE $j$ as
\begin{equation*}
R_{ij} = \sum_{k=1}^K \log\left(1 + \text{SINR}_{ij}[k] \right).
\end{equation*}

Each UE requires a certain level of QoS from the SCBSs. We denote the QoS \emph{required} by UE $j$ from SCBS $i$ by $\overline{\tau}_{ij}$. Depending on the type of the user $j$, from the standpoint of SCBS $i$, the value $\overline{\tau}_{ij}$ may or may not be known by that SCBS. We denote the actual QoS \emph{provided} by SCBS $i$ to UE $j$ by $\tau_{ij}$. We assume that $\tau_{ij} = f_i(u_{ij})$, where $f_i:\realnonnegative \to \realnonnegative$ is a one-to-one, continuously differentiable function that maps the power allocation $u_{ij}$ to a certain value $\tau_{ij}$ of the QoS requirement. We let $\overline{\tau}_{ij} = f_i(\overline{u}_{ij})$, $\overline{u}_{ij} \in \realnonnegative$. 

We assume that each SCBS knows the exact QoS requirements, $\overline{\tau}_{ij}$'s, that its frequent users need. This information is learned over time due to the regular and periodic behavior of these known, frequent users. Whether $\tau_{ij}$ is equal to $\overline{\tau}_{ij}$ (or equivalently, $u_{ij}$ is equal to $\overline{u}_{ij}$) or not for frequent users will be be decided via an optimization problem to be presented next. For the occasional users, we assume that the QoS requirements are not known by the SCBSs due to the lack of interaction between those users and the SCBSs. 

Let $c_{ij}(\tau_{ij},\overline{\tau}_{ij})$ be the cost associated with allocating $\tau_{ij}$ to the $j$-th UE whose QoS requirement is given by $\overline{\tau}_{ij}$. This cost function can take many forms and must be designed so as to penalize prospective misspending of resources, as previously discussed. Hereinafter, without loss of generality, we choose the following continuous cost function:
\begin{eqnarray*} \label{contDist}
c_{ij}(\tau_{ij},\overline{\tau}_{ij}) = \left| \tau_{ij} - \overline{\tau}_{ij} \right|,
\end{eqnarray*}
which is zero only when the SCBS matches the exact QoS demand of the user. If $j \in \mN_2^i$, the cost function is not well defined because the SCBS does not have information about the QoS requirements of occasional users. As a convention, we set $c_{ij} = 0$ for all $j \in \mN_2^i$, for all $i \in \mM$.

We associate a utility function with each SCBS. In particular, the utility function of the $i$-th SCBS is given by:
\begin{equation} \label{cxtUtility}
J_i(U_i,U_{-i} ) = \sum_{j \in \mN_i} R_{ij} - \eta_i\sum_{j \in \mathcal{N}_1^i} c_{ij}(\tau_{ij},\overline{\tau}_{ij}),
\end{equation}
where $U_{-i}$ is the collection of the power allocation matrices of the opponents of SCBS $i$, and $\eta_i > 0$ is a constant controlled by the SCBS depending on what it favors more: maximizing rate or minimizing cost. This constant will be referred to as the ``tradeoff constant" hereinafter.

We can now formulate the problem of resource allocation as a strategic game $\Xi:=\{\mathcal{M}, \{\mathcal{A}_i\}_{i\in\mathcal{M}}, \{U_i\}_{i\in\mathcal{M}}\}$, which is defined by its three main components: (i)- the \emph{players} being the SCBSs in $\mathcal{M}$, (ii)- the \emph{action} sets $\mathcal{A}_i$ of the players defined in (\ref{eqn::actionSet}), and (iii)- the \emph{utility functions} $J_i$ of the players defined in (\ref{cxtUtility}), which capture the gains and costs from each resource allocation decision. In this game, the optimization problem to be solved by the $i$-th SCBS is
\begin{equation*} \label{cxtGame}
\max\limits_{U_i \in \mA_i} J_i(U_i,U_{-i} ), \quad \text{for each fixed } U_{-i}.
\end{equation*}

We are interested in the pure-strategy Nash equilibrium (PSNE) solution of this game.
\begin{definition}\longthmtitle{Pure-strategy Nash equilibrium~\cite{BasarOlsder}}
The $n$-tuple $\{U_1^\star, \hdots, U_n^\star \}$, with $U_i^\star \in \mA_i$, constitutes a PSNE if, for all $ i \in \mM $, the inequality
\[
J_i(U_i^\star,U^\star_{-i}) \geq J_i(U_i,U^\star_{-i})
\]
is satisfied for all $U_i \in \mA_i$.
\end{definition}
According to this definition, no SCBS has an incentive to unilaterally deviate from the person-by-person optimal solution $\{U_1^\star, \hdots, U_n^\star \}$.

Before we state the main result of this paper, we introduce the following definitions:
\begin{eqnarray}
&& \hspace{-7mm} \subscr{\sigma^2}{min} := \min_{j\in \mN, k\in \mK} \sigma_j^2[k],  \hspace{3mm}\subscr{\beta}{min} := \min_{\substack{i\in \mM, j\in \mN\\ k\in \mK}} |h_{ij}[k]|^2d_{ij}^{-\alpha},\label{eqn::defs1}\\
&& \hspace{-7mm} \subscr{\sigma^2}{max} := \max_{j\in \mN, k\in \mK} \sigma_j^2[k] , \hspace{3mm} \subscr{\beta}{max} :=  \max_{\substack{i\in \mM, j\in \mN\\ k\in \mK}} |h_{ij}[k]|^2d_{ij}^{-\alpha}. \label{eqn::defs2}
\end{eqnarray}
Further, let $\subscr{K}{max}$ denote the largest number of subcarriers allocated to a UE by an SCBS, and define
\[
\xi_1 :=  \frac{\subscr{\beta}{min} (\subscr{\sigma^2}{min})^3}{(M-1)\subscr{K}{max}\subscr{\beta}{max}^3}, \quad \xi_2 := \frac{1-\subscr{\sigma^2}{max} }{M\subscr{\beta}{max} }.
\]
The following theorem establishes the existence and uniqueness of PSNE for $\Xi$.
\begin{theorem} \label{thm::uniquePSNE} \longthmtitle{Existence and uniqueness of PSNE for the context-aware power allocation game}
Assume that the cost function $c_{ij}$ is linear in $u_{ij}$, for all $i\in \mM$ and $j \in \mN_i$, and that $\subscr{\sigma^2}{max} < 1$. Then, the proposed context-aware power allocation game $\Xi$ admits a unique PSNE if
\begin{equation}\label{eqn::Pmax}
\subscr{P}{max} < \min\left\{ \xi_1, \xi_2    \right\},
\end{equation}
\end{theorem}
\begin{proof}
See the Appendix.
\end{proof}

\begin{remark} \longthmtitle{Independence from the number of UEs}
An attractive feature of the condition on the maximum transmit power provided by Theorem \ref{thm::uniquePSNE} is that it does not depend on the number of the UEs in the network. Hence, this makes the existence and uniqueness of PSNE guaranteed for a large class of networks in which the number of UEs is large. \oprocend
\end{remark}

\begin{remark}\label{rem::nonlinearQoS}\longthmtitle{Nonlinear QoS requirements}
The linearity assumption on $c_{ij}$ guarantees that $J_i$ is concave in $U_i$, which in turn guarantees the existence of a PSNE as stated in the proof of Theorem \ref{thm::uniquePSNE}. It is possible to extend the theorem above to the case where $c_{ij}$ is not linear in $u_{ij}$. This will entail placing conditions on $\eta_i$ in order to preserve the concavity of $J_i$ in $U_i$. In particular, we would need $\frac{\partial^2}{\partial U_i^2} J_i \preceq 0$, which is equivalent to
\[
\sum_{j \in \mN_i}\frac{\partial^2}{\partial U_i^2}  R_{ij} - \eta_i\sum_{j \in \mathcal{N}_1^i} \frac{\partial^2}{\partial U_i^2} c_{ij}(\tau_{ij},\overline{\tau}_{ij}) \preceq 0,
\] 
for all $i\in \mM$. \oprocend
\end{remark}

The following corollary proposes a method to compute the unique PSNE of $\Xi$. As proposed by Rosen in \cite{Rosen1965}, we allow the players to update their strategies dynamically at a rate that is proportional to the gradient of their utility functions, while satisfying the power constraints. To this end, let $\Phi_i: \mathbb{R}^{N\times K} \to \mA_i$ be the projection operator that guarantees that the power constraint is satisfied for each player.
\begin{corollary}\label{cor::compute}\longthmtitle{Computation of the unique PSNE}
Assume that  the cost function $c_{ij}$ is linear in $u_{ij}$, for all $i\in \mM$ and $j\in \mN_i$, and let $\subscr{P}{max}$ be chosen according to (\ref{eqn::Pmax}). Then, the dynamical system
\begin{equation*}
\frac{d}{dt}U_i = \Phi_{i} \left(  \frac{\partial}{\partial U_i} J_i     \right), \quad i \in \mM,
\end{equation*}
is globally asymptotically stable and $\{U_1,\hdots,U_M \}$ converges to the unique PSNE of the context-aware power allocation game, $\{U_1^\star,\hdots,U_M^\star \}$, starting from any arbitrary initial condition $\{U_1(0),\hdots,U_M(0) \}$.
\end{corollary}
\begin{proof}
Under the assumptions made in the statement of the corollary, we have $G(u)+G(u)^T \prec 0$, for all $u \in \hat{\mA}$, where $G(u)$ and $\hat{\mA}$ are defined in the Appendix. The corollary then follows by Theorem 9 in \cite{Rosen1965}.
\end{proof}



\section{Simulation Results}\label{sec:sim}
For simulations, consider a two-tier $500$~m $\times$ $500$~m square network in which the UEs and SCBSs are deployed uniformly. Here, the SCBSs represent operator-owned, outdoor picocells. The noise power is set to $-110$ dBm. We let $\alpha = 3$. The maximum transmit power of each picocell SCBS is chosen so as (\ref{eqn::Pmax}) is satisfied. We compute the PSNE of $\Xi$ using the algorithm proposed in Corollary \ref{cor::compute}. We are interested in simulating the worst-case scenario in which each SCBS serves all the users in the network, and that all the SCBSs serve a given UE over the same subcarrier. Formally, we assume that $K=N$, $\mN_i = \mN$, $\mK_i=\mK$, and $\pi_i(k) = \pi_j(k)$ for all $i\in \mM$ and $k \in \mK$.

We will compare the performance of our context-aware allocation to two widely used allocation schemes: sum-rate maximization and proportionally fair (PF) allocation. The sum-rate maximization problem is $\max_{U_i \in \mA_i} \sum_{j \in \mathcal{N}_i} R_{ij}$. As in our allocation scheme, we let the agents play a noncooperative game.  Clearly, the sum-rate maximization problem need not yield a fair allocation. For this reason, we also compare our scheme to the popular proportionally fair allocation scheme which overcomes this problem at the expense of possibly worse sum-rate compared to that provided by the sum-rate maximization scheme. The PF allocation is the solution to the following problem: $\max_{U_i \in \mA_i} \sum_{j \in \mathcal{N}_i} \log(R_{ij})$. As in our allocation scheme, we let the SCBSs engage in a noncooperative game in all schemes. We compare the performance of the three schemes by evaluating the utility functions of each SCBS as given by (\ref{cxtUtility}). We assume that the QoS requirements of the frequent UEs are certain desired power allocation levels. To this end, we will set the QoS mappings for frequent users to $\tau_{ij} = f_i(u_{ij})=u_{ij}$, and therefore we have $c_{ij} = |u_{ij} - \overline{u}_{ij}|$, for all $i\in \mM$ and $j \in \mN_1^i$. The true QoS requirements of frequent users $\overline{\tau}_{ij} = \overline{u}_{ij}$ are generated at random. All statistical results are averaged over all SCBS and UE locations via a large number of independent runs.

In Fig. \ref{fig::avgUtilNumFreq}, we compare the average utility  per SBCS resulting from all the three schemes for a network with $M=5$ SCBSs, $5$ occasional users, and a varying number of frequent users. For all SCBSs, the tradeoff constants are set to $\eta_i = 2,\ \forall i \in \mathcal{M}$. Fig.~\ref{fig::avgUtilNumFreq} shows that, as the number of frequent users increases, the performances of all three schemes improve due to the additional users. Particularly, for the proposed context-aware approach, this increase is due to the fact that this approach is capable of exploiting the knowledge of the traffic patterns of a larger number of users (the frequent ones) in the network  hence yielding an improved utility. In Fig.~\ref{fig::avgUtilNumFreq}, we can see that the proposed context-aware game has a significant performance improvement, in terms of the average utility per SCBSs, compared to both PF allocation and sum-rate maximization for any number of frequent users.

\begin{figure}[!t]
  \begin{center}
    \includegraphics[width=8cm]{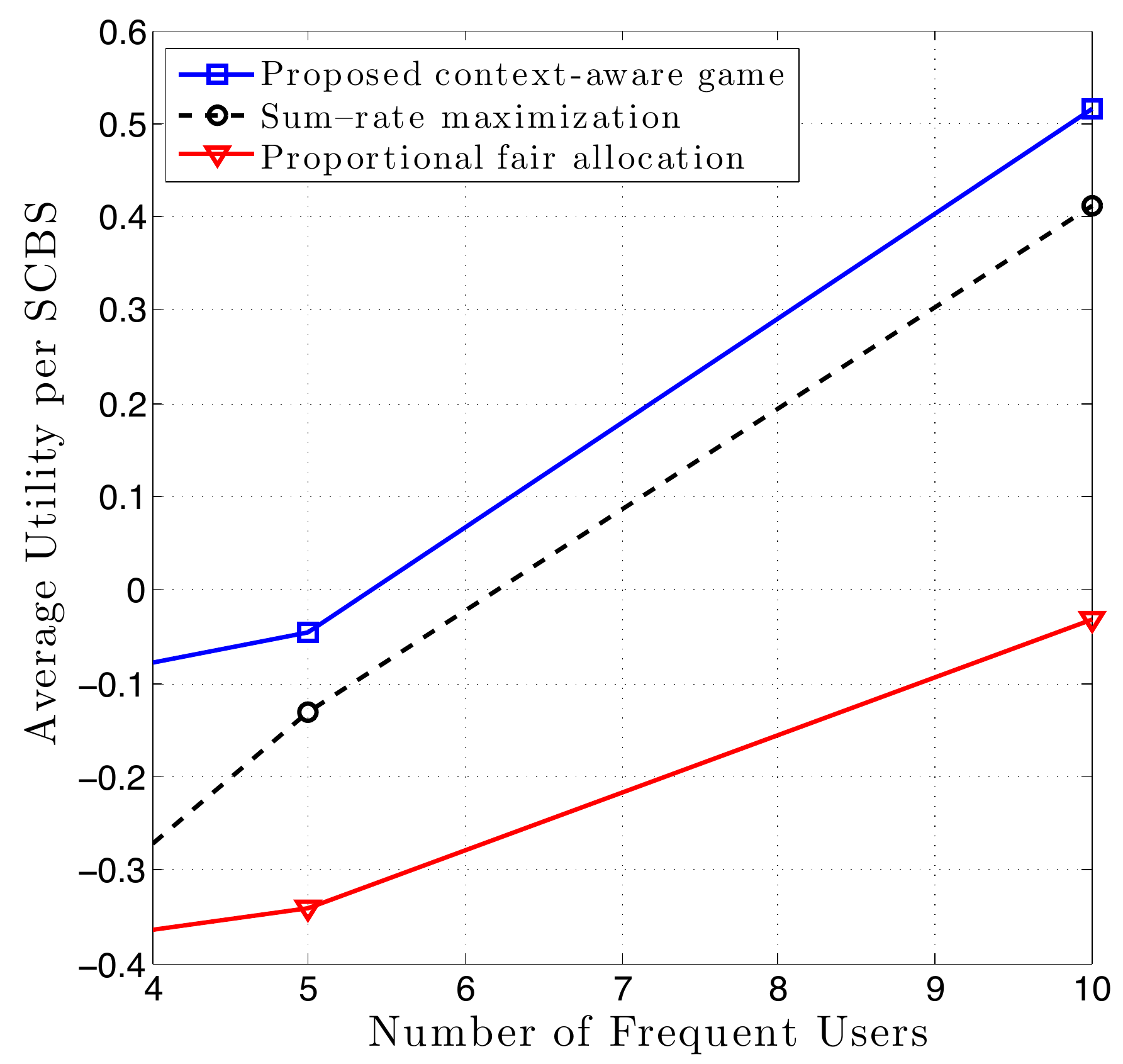}
    \caption{\label{fig::avgUtilNumFreq} \small Average utility per SCBS as the number of frequent users varies, for a network with $5$ SCBSs and $5$ occasional users.}
  \end{center}
\end{figure}

In Fig. \ref{fig::avgUtilTradeOff}, we study the impact of the tradeoff constants $\eta_i$ on the average utility for a network with $6$ frequent users, $5$ SCBSs, and $2$ occasional users. We assume that all SCBSs have the same tradeoff constant, i.e., $\eta_i = \eta$. Fig. \ref{fig::avgUtilTradeOff} shows that our proposed context-aware scheme outperforms sum-rate maximization by a significant margin for all $\eta$, reaching up to $\%56$ for $\eta = 30$.  This gain mainly stems from the fact that sum-rate maximization does not take the misspending cost into consideration and its performance will be inferior for networks that place emphasis on optimizing this cost. Similarly, Fig.~\ref{fig::avgUtilTradeOff} shows that the context-aware scheme outperforms proportionally fair allocation for all values of $\eta$. 
\begin{figure}[!t]
  \begin{center}
    \includegraphics[width=8cm]{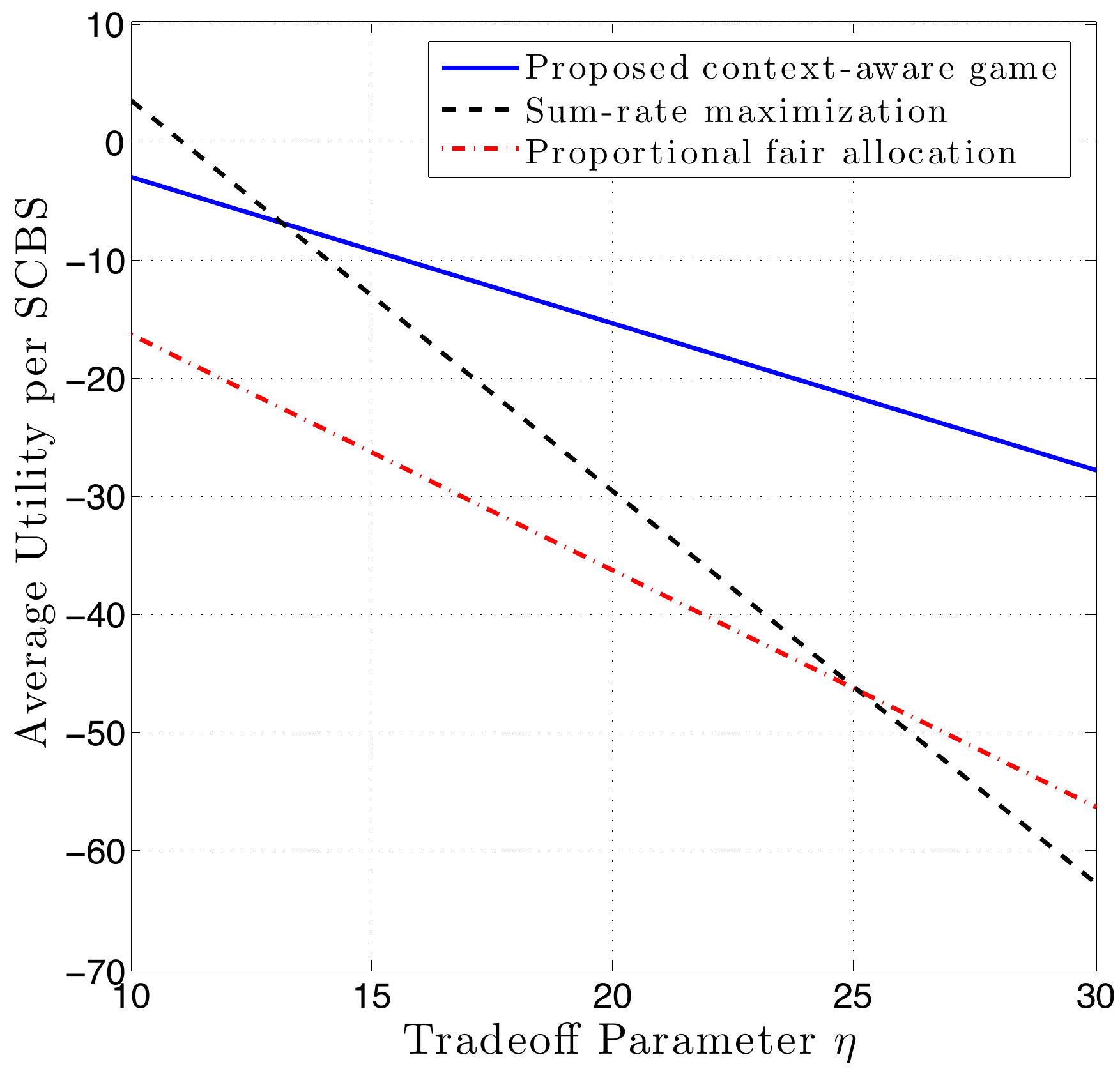}
    \caption{\label{fig::avgUtilTradeOff} \small Average utility per SCBS as the tradeoff parameter $\eta$ varies, for $5$ SCBSs, $6$ frequent users, and $2$ occasional users.}
  \end{center}
\end{figure}

\section{Conclusions} \label{sec:conc}
In this paper, we have introduced a novel, context-aware approach for resource allocation in small cell networks that is built on the idea of categorizing the wireless users into two key types: frequent users, who are regular users of a given small cell, and occasional users, who are only passing users. Then, we have formulated the problem as a noncooperative game via which the SCBSs aim to allocate power based on their knowledge of the traffic patterns of their frequent users. In particular, the SCBSs optimize the tradeoff between maximizing the users' rate and minimizing the cost of misspending resources that is associated with under or over matching the frequent users' demands. We have shown the existence of a unique pure strategy Nash equilibrium and we have studied its properties. Simulation results have shown promising performance gains that context-aware schemes can provide in small cell networks. This paper only scratches the surface of an emerging area of context-aware resource allocation in small cell networks. Motivated by the proximity of the users to their small cells, we expect that future small cell networks are able to exploit a significant amount of information knowledge on their users so as to better service them and optimize the overall resource usage in the network.

\def\baselinestretch{0.8}
\bibliographystyle{IEEEtran}
\bibliography{refs}

\appendix

\subsection{Proof of Theorem~\ref{thm::uniquePSNE}}
We start by proving the case where $K=N$ and each SCBS allocates only one subcarrier to a given user. We will proceed by invoking the existence and uniqueness results of \cite{Rosen1965}. Under the linearity assumption on $f_i$, it follows that the utility of the $i$-th player, $J_i$, is concave in $U_i$ for each fixed value of $U_{-i}$. Hence, it follows by Theorem 1 in \cite{Rosen1965} that a PSNE exists.

Before we prove the uniqueness of PSNE, we introduce some notation to streamline our analysis. Given a certain subcarrier allocation at SCBS $i$, by construction of $\mA_i$, it follows that $U_i$ will have a total of $N_i$ nonzero entries. We collect these entries into a vector $u_i \in \mathbb{R}^{N_i}$, and we denote its entries by $u_i = [u_{i1},\hdots,u_{iN_i}]^T$, where $u_{ij} = u_{i\pi_i(k_j^i)}[k_j^i]$ and $k_j^i \in \mK_i$ is the subcarrier allocated by SCBS $i$ to user $j \in \mN_i$. Once the subcarrier allocation has been fixed, we can obtain $u_i$ from $U_i$ via a bijective mapping. Define the following restricted action set
\[
\hat{\mA}_i =\left\{u_i \in \realnonnegative^{N_i} : u_{ij} = u_{i\pi_i(k_j)}[k_j]; \mathbf{1}^Tu_i  \leq \subscr{P}{max}  \right\},
\]
and let $\hat{\mA} := \hat{\mA}_1 \times \hdots \times \hat{\mA}_M$. Further, let $u = [u_1^T,\hdots,u_M^T]^T$, $u_i \in \hat{\mA}_i$, and define the following matrix
\[
G(u) :=
\left[
\begin{array}{ccc}
\frac{\partial^2}{\partial u_1^2}J_1   & \hdots & \frac{\partial^2}{\partial u_1 \partial u_{M}}J_1  \\
\vdots & \ddots & \vdots \\
\frac{\partial^2}{\partial u_1 \partial u_M}J_M & \hdots & \frac{\partial^2}{\partial u_M^2}J_M
\end{array}
\right].
\]

To prove uniqueness, we will invoke Theorem 6 in \cite{Rosen1965}, which states that the negative definiteness of the matrix $G(u)+G(u)^T$, for all $u \in \hat{\mA}$, is a sufficient condition for the uniqueness of PSNE. Note that the diagonal entries of $\frac{\partial^2}{\partial u_i^2}J_i$ are negative, for all $i$, because $J_i$ is concave in $u_{ij}$, for all $j\in\mN_i$. Hence, if the matrix $G(u)+G(u)^T$ is diagonally dominant for all $u \in \hat{\mA}$, it follows that $G(u)+G(u)^T$ is a negative definite matrix for all $\hat{\mA}$. Formally, for $G(u)+G(u)^T$ to be diagonally dominant, we must have
\begin{eqnarray*}
2\left|\frac{\partial^2}{\partial u_{il}^2}J_i \right| & > & 2  \sum_{\substack{h\in \mN_i\\  h\neq l}} \left| \frac{\partial^2}{\partial u_{il} \partial u_{ih}}J_i \right| \\
& + & \sum_{\substack{j \in \mM \\ j \neq i}} \sum_{h\in \mN_i} \left| \frac{\partial^2}{\partial u_{il} \partial u_{jh}}J_i + \frac{\partial^2}{\partial u_{il} \partial u_{jh}}J_j \right| ,
\end{eqnarray*}
for all $l \in \mN_i$ and all $i \in \mM$. Note that for $h \neq l$, $l,h \in \mN_i$, we have
\[
\frac{\partial^2}{\partial u_{il} \partial u_{ih}}J_i = \frac{\partial}{\partial u_{ih}}\left(\frac{\partial}{\partial u_{il}}(R_{il} -  \eta_i c_{il}) \right) = 0,
\]
because of our assumption that there is no inter-user interference. Similarly, for $h \neq l$, $l,h \in \mN_i$, we have
\begin{eqnarray*}
\frac{\partial^2}{\partial u_{il} \partial u_{jh}}J_i = \frac{\partial}{\partial u_{jh}}\left(\frac{\partial}{\partial u_{il}}(R_{il} -  \eta_i c_{il}) \right) = 0,
\end{eqnarray*}
which again follows because there is no inter-user interference, by assumption. Using these, we can now simplify the diagonal dominance condition to requiring
\begin{equation} \label{eqn::diagdom}
2\left|\frac{\partial^2}{\partial u_{il}^2}J_i \right| > \sum_{\substack{j \in \mM \\ j \neq i}}  \left| \frac{\partial^2}{\partial u_{il} \partial u_{jl}}J_i + \frac{\partial^2}{\partial u_{il} \partial u_{jl}}J_j \right|,
\end{equation}
for all $l \in \mN_i$ and all $i \in \mM$.

Next, we will find bounds on the terms involved in (\ref{eqn::diagdom}). Let $\beta_{ij} := |h_{ij}[k_j^i]|^2d_{ij}^{-\alpha}$ and $I_{ij} = \sum_{l\neq i} h_{lj}[k_j^l]|^2d_{lj}^{-\alpha}u_{lj}[k_j^l]$, where $k_j^l\in \mK_i$ is the subcarrier allocated by SCBS $i$ to user $j\in\mN_i$. Note that
\[
\frac{\partial}{\partial u_{il}} R_{il} =  \beta_{il} \frac{ I_{il} + \sigma_l^2}{I_{il} +\beta_{il}u_{il} + \sigma_l^2}
\]
Using this, and recalling that $c_{ij}$ is linear in $u_{ij}$, by assumption, we can write
\[
\frac{\partial^2}{\partial u_{il}^2} J_i = \frac{\partial^2}{\partial u_{il}^2} R_{il} = - \beta_{il}^2 \frac{ I_{il} + \sigma_l^2}{(I_{il} +\beta_{il}u_{il} + \sigma_l^2)^2}.
\]
Using the definitions in (\ref{eqn::defs1}) and (\ref{eqn::defs2}), we can then write
\begin{equation} \label{eqn::bound1}
\left| \frac{\partial^2}{\partial u_{il}^2} J_i \right| \geq \frac{ \subscr{\beta}{min}\subscr{\sigma^2}{min}}{(M\subscr{\beta}{max}\subscr{P}{max} + \subscr{\sigma^2}{max})^2}.
\end{equation}
Assuming that SCBS $i$ and $j$ interfere at user $l$, i.e., $u_{il}[k_l^i]$, $u_{jl}[k_l^j$  are both positive and $k_l^i=k_l^j$, we can write
\begin{eqnarray*}
\frac{\partial^2}{\partial u_{il} \partial u_{jl}} J_i & = & \frac{\partial^2}{\partial u_{il} \partial u_{jl}} R_{il}\\
& = & \beta_{il}\frac{\beta_{jl} (I_{il} +\beta_{il}u_{il} + \sigma_l^2) - \beta_{jl}(I_{il} + \sigma_l^2)}{(I_{il} +\beta_{il}u_{il} + \sigma_l^2)^2} \\
& = &\frac{\beta_{jl} \beta_{il}^2u_{il} }{(I_{il} +\beta_{il}u_{il} + \sigma_l^2)^2}.
\end{eqnarray*}
Using this and the triangle inequality, we can now obtain the following bound:
\begin{equation} \label{eqn::bound2}
\left| \frac{\partial^2}{\partial u_{il} \partial u_{jl}}J_i + \frac{\partial^2}{\partial u_{il} \partial u_{jl}}J_j\right| \leq 2\frac{\subscr{\beta}{max}^3\subscr{P}{max} }{(\subscr{\sigma^2}{min})^2}.
\end{equation}
Using (\ref{eqn::bound1}) and (\ref{eqn::bound2}), it follows that a sufficient condition for (\ref{eqn::diagdom}) to be satisfied is:
\[
\frac{\subscr{\beta}{min} \subscr{\sigma^2}{min}}{(M\subscr{\beta}{max}\subscr{P}{max} + \subscr{\sigma^2}{max})^2} > (M-1)\frac{\subscr{\beta}{max}^3\subscr{P}{max} }{(\subscr{\sigma^2}{min})^2}
\]
or
\begin{equation} \label{eqn::diagdom2}
\subscr{P}{max} (M\subscr{\beta}{max}\subscr{P}{max} + \subscr{\sigma^2}{max})^2 < \frac{ \subscr{\beta}{min}(\subscr{\sigma^2}{min})^3}{(M-1)\subscr{\beta}{max}^3}
\end{equation}
Then, assuming that $\subscr{\sigma^2}{max} < 1$, inequality (\ref{eqn::diagdom2}) can be satisfied by choosing $\subscr{P}{max} $ as in (\ref{eqn::Pmax}), where $\subscr{K}{max} = 1$ in this case, and we recall that $\subscr{K}{max}$ denotes the maximum number of subcarriers allocated by an SCBS to a user. This completes the proof for the case where $K=N$ and each SCBS allocates only one subcarrier to a given user.

For the case when $K \geq N$, and the SCBSs are allowed to allocate multiple subcarriers to a UE, let us first denote the number of subcarriers allocated by SCBS $i$ to user $j\in\mN_i$ by $K_j^i$. In this case, the dimension of vector $u_i$ would be $\sum_{j \in \mN_i} K_j^i$. The vector $u_i$ can be obtained from $U_i$ using a similar procedure to the one described above. Using the no inter-user interference assumption, we conclude that the only change to the above steps is the inclusion of a summation over the subcarriers allocated by SCBS $j$ to user $l\in\mN_j$ on the right hand side of condition (\ref{eqn::diagdom}). To capture this extra summation, we must change the inequality in (\ref{eqn::diagdom2}) to
\begin{equation*} 
\subscr{P}{max} (M\subscr{\beta}{max}\subscr{P}{max} + \subscr{\sigma^2}{max})^2 < \frac{\subscr{\beta}{min} (\subscr{\sigma^2}{min})^3}{(M-1)\subscr{K}{max}\subscr{\beta}{max}^3}.
\end{equation*}
Similar to the previous case, the above inequality can be satisfied by choosing $\subscr{P}{max} $ as in (\ref{eqn::Pmax}), and the proof is complete.

\end{document}